\documentclass[11pt,a4paper]{article}

\usepackage[utf8]{inputenc} 
\usepackage[T1]{fontenc}    
\usepackage{lmodern}        

\usepackage{amsmath, amssymb, amsthm} 
\usepackage{graphicx}       
\usepackage[margin=1in]{geometry} 
\usepackage[numbers]{natbib} 
\usepackage{enumitem}       
\usepackage{listings}       
\usepackage{textcomp}       
\usepackage{booktabs}       

\newtheorem{theorem}{Theorem}[section]
\newtheorem{lemma}[theorem]{Lemma}

\theoremstyle{definition}
\newtheorem{definition}[theorem]{Definition}
\newtheorem{assumption}[theorem]{Assumption}
\theoremstyle{remark}
\newtheorem{remark}[theorem]{Remark}

\newcommand{\KL}{\ensuremath{\mathrm{KL}}} 

\title{\textbf{A Hierarchical Decomposition of Kullback-Leibler Divergence: Disentangling Marginal Mismatches from Statistical Dependencies}} 
\author{
    William Cook
}
\date{\today} 

\begin{document}
\maketitle

\begin{abstract}
The Kullback-Leibler (KL) divergence is a foundational measure for comparing probability distributions, yet in multivariate settings, its structure is often opaque—conflating marginal mismatches and statistical dependencies. We derive an algebraically exact, additive, and hierarchical decomposition of the KL divergence ($\KL(P_k \| Q^{(\otimes k)})$) between a joint distribution ($P_k$) and a product reference ($Q^{(\otimes k)}$). The total divergence splits into the sum of marginal KLs, ($\sum_{i=1}^k \KL(P_i \| Q)$), and the total correlation ($C(P_k)$), which we further decompose as ($C(P_k) = \sum_{r=2}^k I^{(r)}(P_k)$) using Möbius inversion on the subset lattice. Each ($I^{(r)}$) quantifies the distinct contribution of ($r$)-way statistical interactions to the total divergence. This yields a decomposition of this form that is both algebraically complete and interpretable using only standard Shannon quantities, with no approximations or model assumptions. Numerical validation using hypergeometric sampling confirms exactness to machine precision across diverse system configurations. This framework enables precise diagnosis of divergence origins—marginal vs. interaction—across applications in machine learning, econometrics, and complex systems.
\end{abstract}

\section{Introduction}
\label{sec:intro}

Understanding why high-dimensional probability distributions differ—beyond merely quantifying how much—is a central challenge in statistics and machine learning. The Kullback-Leibler (KL) divergence, or relative entropy \cite{CoverThomas2006}, provides a foundational measure of the distinction between a `true' distribution $P$ and a model or reference distribution $Q$. Its applications are widespread, spanning machine learning (e.g., variational inference, model selection) \cite{Bishop2006}, statistical physics, and computational biology \cite{Presse2013}.

However, in multivariate systems described by a joint distribution $P_k(X_1, \dots, X_k)$ compared to a simpler reference such as a product distribution $Q^{(\otimes k)} = \prod_{i=1}^k Q(X_i)$, where each factor is identical to a fixed reference distribution $Q$, the total KL divergence $\KL(P_k \| Q^{(\otimes k)})$ yields a single scalar value. This aggregate measure conflates two distinct sources of divergence: mismatches between each individual marginal distribution $P_i$ and the reference $Q$, and statistical dependencies (interactions) among the variables $X_1, \dots, X_k$ within $P_k$ that are absent in the independent reference $Q^{(\otimes k)}$.

Disentangling these components is crucial for interpreting the sources of divergence or model misfit. While established measures like mutual information and total correlation \cite{Watanabe1960} quantify aspects of statistical dependency, they do not typically offer an exact decomposition of the KL divergence itself relative to $Q^{(\otimes k)}$. Nor do they systematically partition the divergence into contributions from marginal mismatches versus a hierarchy of interaction orders. This leaves key diagnostic questions unanswered: In machine learning, does a model's poor fit (high KL divergence) stem primarily from misrepresenting individual feature distributions (marginal effects), or from failing to capture complex feature interactions (dependency effects)? In systems biology, does divergence from a baseline biological state arise from individual component dysregulation or from altered network interactions?

This paper presents an exact, additive decomposition of $\KL(P_k \| Q^{(\otimes k)})$ that directly addresses this challenge. We demonstrate algebraically that:
\begin{equation*}
\KL(P_k \| Q^{(\otimes k)}) = \underbrace{\left[ \sum_{i=1}^k \KL(P_i \| Q) \right]}_{\text{Sum of Marginal Divergences}} + \underbrace{\left[ \sum_{r=2}^k I^{(r)}(P_k) \right]}_{\substack{\text{Total Correlation /} \\ \text{Sum of Interactions ($r \ge 2$)}}}
\end{equation*}
Here, the first term isolates the divergence due solely to differences between each marginal $P_i$ and the reference $Q$. The second term is precisely the total correlation $C(P_k)$, which encapsulates the contributions from all statistical dependencies within $P_k$. Furthermore, we show this total correlation term decomposes hierarchically into the sum of $r$-way interaction information terms $I^{(r)}(P_k)$ for $r=2, \dots, k$.

The derivation leverages the fundamental properties of Shannon entropy and its relationship to interaction information, defined via Möbius inversion on the Boolean lattice of variable subsets \cite{AyEtAl2017, Amari2016}. This approach yields a principled and complete breakdown using only classical information-theoretic quantities.

This decomposition provides a structured framework for dissecting KL divergence, offering fine-grained insights by cleanly separating marginal-level mismatches from the cumulative contributions of pairwise, triplet, and higher-order statistical interactions. It constitutes an exact, additive, and hierarchical decomposition of $\KL(P_k \| Q^{(\otimes k)})$ into these standard information-theoretic components. The framework is validated numerically using multivariate hypergeometric models, confirming its exactness.

Preceding Section~\ref{sec:discussion} is an accessible synthesis of the core contribution for readers without the necessary prior mathematical foundations implicitly assumed elsewhere.

The paper is structured as follows:
\begin{itemize}
    \item Section~\ref{sec:theory}: Theoretical framework. Defines key concepts and rigorously derives the main decomposition theorem (Theorem~\ref{thm:KLDecomp}).
    \item Section~\ref{sec:visuals}: Visual illustrations. Illustrates the decomposition using diagrams and empirical data from numerical cases.
    \item Section~\ref{sec:experiments}: Experiments and numerical validation. Details the numerical validation methodology and presents results confirming the decomposition's exactness.
    \item Section~\ref{sec:intuition}: Intuition and Analogy. Accessible synthesis of the core contribution.
    \item Section~\ref{sec:discussion}: Discussion. Interprets the findings, discusses potential applications, connections to information geometry, and limitations.
    \item Section~\ref{sec:conclusion}: Conclusion. Summarises the contribution and outlines future research directions.
\end{itemize}

\section{Theoretical Framework}
\label{sec:theory}

\subsection{Definitions and Preliminaries}
Let $X = (X_1, \dots, X_k)$ be a discrete random vector over a finite alphabet $\mathcal{X}$, with joint probability distribution $P_k(x_1, \dots, x_k)$. Let $Q$ be a reference probability distribution on $\mathcal{X}$.

\begin{assumption}
\label{ass:placeholder}
We assume $Q(x) > 0$ for all $x$ in $\mathcal{X}$. This ensures that the logarithms in the KL divergence definitions are well-defined.
\end{assumption}

Define the product reference distribution:
\begin{equation*}
Q^{(\otimes k)}(x_1, \dots, x_k) = \prod_{i=1}^k Q(x_i)
\end{equation*}

We define $Q$ as a fixed reference marginal distribution over the support $\mathcal{X}$.
The product reference distribution is then constructed by taking the independent product of $k$ identical marginals:
\begin{equation*}
Q^{(\otimes k)}(x_1, \dots, x_k) := \prod_{i=1}^k Q(x_i).
\end{equation*}
That is, we evaluate the divergence of $P_k$ relative to an i.i.d. reference model in which each variable $X_i \sim Q$ independently.

For any subset $S \subseteq [k] = \{1, \dots, k\}$, let $X_S = \{X_i : i \in S\}$. Denote:
\begin{itemize}
    \item $P_S(x_S) = \sum_{x_{[k] \setminus S}} P_k(x_1, \dots, x_k)$, the marginal distribution of $X_S$. Notably, $P_i(x_i)$ is the marginal distribution of $X_i$.
    \item $H(X_S) = - \sum_{x_S} P_S(x_S) \log_2 P_S(x_S)$, the Shannon entropy of $X_S$ in bits, with the standard conventions $H(X_\emptyset) = 0$ and $0 \log_2 0 = 0$.
\end{itemize}

\begin{definition}[Interaction Information]
\label{def:interactioninfo}
The \textbf{interaction information} for a non-empty subset $S \subseteq [k]$ is defined via Möbius inversion of the negentropy function on the subset lattice:
\begin{equation*}
I(S) = - \sum_{T \subseteq S} (-1)^{|S| - |T|} H(X_T)
\end{equation*}
This aligns with information geometry conventions \cite{AyEtAl2017, Amari2016}. Intuitively, $I(S)$ measures the synergy or redundancy within the subset $S$—the portion of information generated by the variables acting together that cannot be attributed to any smaller sub-group.
\end{definition}

\begin{remark}[Sign Convention Warning]
\label{rem:signconvention}
The interaction information $I(S)$ is defined via Möbius inversion with alternating signs, following conventions from \cite{AyEtAl2017}.
This implies that for subsets $S$ with odd cardinality $|S|$, $I(S)$ is typically negative, and for even $|S|$, positive.
For instance:
\begin{itemize}
    \item $I(\{i\}) = -H(X_i)$
    \item $I(\{i,j\}) = I(X_i; X_j) \ge 0$ (standard mutual information)
    \item $I(\{i,j,k\})$ may be negative (synergy) or positive (redundancy). This is the negative of the quantity sometimes called co-information.
\end{itemize}
This alternating-sign convention is crucial for the algebraic consistency in the decomposition (Theorem~\ref{thm:KLDecomp}). It differs from conventions where all terms are defined non-negatively (e.g., \cite{WilliamsBeer2010}).
\end{remark}

\begin{remark}[Sign Convention and Examples]
\label{rem:SignConventionExamples}
The interaction information $I(S)$ varies in sign interpretation depending on $|S|$:
\begin{itemize}
    \item $|S| = 1$: $I(\{i\}) = -H(X_i)$, the negative marginal entropy.
    \item $|S| = 2$: $I(\{i,j\}) = H(X_i) + H(X_j) - H(X_{i,j})$, which is the standard mutual information $I(X_i; X_j) \ge 0$.
    \item $|S| = 3$: $I(\{i,j,k\}) = -[H(X_{i,j,k}) - H(X_{i,j}) - H(X_{i,k}) - H(X_{j,k}) + H(X_i) + H(X_j) + H(X_k)]$. This is the negative of the standard co-information or interaction information measure \cite{McGill1954}. A positive $I(\{i,j,k\})$ in our convention indicates redundancy, while a negative value indicates synergy.
\end{itemize}
The signs alternate based on $|S|$’s parity relative to mutual information. This sign convention, consistent with \cite{AyEtAl2017, Amari2016}, ensures that $I(\{i,j\})$ matches the standard non-negative mutual information $I(X_i; X_j)$, while allowing the decomposition (Theorem~\ref{thm:KLDecomp}) to hold algebraically across all orders.
\end{remark}

\begin{definition}[Total $r$-way Interaction Information]
\label{def:TotalInteractionInfo}
The \textbf{total $r$-way interaction information} is the sum of interaction information over all subsets of size $r$:
\begin{equation*}
I^{(r)}(P_k) = \sum_{S \subseteq [k], |S|=r} I(S) \quad \text{for } r = 1, \dots, k.
\end{equation*}
Based on the sign convention:
\begin{itemize}
    \item $I^{(1)}(P_k) = \sum_{i=1}^k I(\{i\}) = - \sum_{i=1}^k H(X_i)$.
    \item $I^{(2)}(P_k) = \sum_{1 \le i < j \le k} I(X_i; X_j)$, the sum of all pairwise mutual informations.
    \item $I^{(3)}(P_k)$ sums all triplet interactions $I(\{i,j,k\})$, and so on.
\end{itemize}
\end{definition}

\begin{definition}[Total Correlation]
\label{def:TotalCorrelation}
The \textbf{total correlation} (or multi-information) \cite{Watanabe1960} quantifies the total amount of statistical dependency among the variables in $X_{[k]}$:
\begin{equation*}
C(P_k) = \sum_{i=1}^k H(X_i) - H(X_{[k]})
\end{equation*}
It measures the redundancy among the variables, or equivalently, the KL divergence between the joint distribution and the product of its marginals: $C(P_k) = \KL(P_k \| \prod_{i=1}^k P_i) \ge 0$.
\end{definition}

\begin{definition}[Kullback-Leibler Divergence]
\label{def:KLDiv}
The \textbf{KL divergence} between the joint distribution $P_k$ and the product reference distribution $Q^{(\otimes k)}$ is:
\begin{equation*}
\KL(P_k \| Q^{(\otimes k)}) = \sum_{x_1, \dots, x_k} P_k(x_1, \dots, x_k) \log_2 \left[ \frac{P_k(x_1, \dots, x_k)}{\prod_{i=1}^k Q(x_i)} \right]
\end{equation*}
The marginal KL divergence between the marginal distribution $P_i$ and the reference distribution $Q$ is:
\begin{equation*}
\KL(P_i \| Q) = \sum_{x_i} P_i(x_i) \log_2 \left[ \frac{P_i(x_i)}{Q(x_i)} \right]
\end{equation*}
\end{definition}

\subsection{Core Identities}
We establish the crucial link between total correlation and the sum of interaction information terms.

\begin{lemma}[Total Correlation from Interaction Information]
\label{lem:TotalCorrelationII}
The sum of interaction information terms of order 2 and higher equals the total correlation:
\begin{equation} \label{eq:TotalCorrFromII}
\sum_{r=2}^k I^{(r)}(P_k) = C(P_k)
\end{equation}
\end{lemma}
\begin{proof}
From Definition~\ref{def:interactioninfo}, the relationship between entropy and interaction information is given by Möbius inversion on the subset lattice. The inverse relation states that for any non-empty subset $S \subseteq [k]$:
\begin{equation*}
-H(X_S) = \sum_{\substack{T \subseteq S \\ T \ne \emptyset}} I(T)
\end{equation*}
Applying this to the full set $S = [k]$:
\begin{equation*}
-H(X_{[k]}) = \sum_{\substack{T \subseteq [k] \\ T \ne \emptyset}} I(T) = \sum_{r=1}^k \sum_{\substack{T \subseteq [k] \\ |T|=r}} I(T) = \sum_{r=1}^k I^{(r)}(P_k)
\end{equation*}
We can split the sum at $r=1$:
\begin{equation*}
-H(X_{[k]}) = I^{(1)}(P_k) + \sum_{r=2}^k I^{(r)}(P_k)
\end{equation*}
Substituting the definition $I^{(1)}(P_k) = - \sum_{i=1}^k H(X_i)$:
\begin{equation*}
-H(X_{[k]}) = - \sum_{i=1}^k H(X_i) + \sum_{r=2}^k I^{(r)}(P_k)
\end{equation*}
Rearranging the terms yields:
\begin{equation*}
\sum_{i=1}^k H(X_i) - H(X_{[k]}) = \sum_{r=2}^k I^{(r)}(P_k)
\end{equation*}
The left-hand side is precisely the definition of total correlation $C(P_k)$ (Definition~\ref{def:TotalCorrelation}). Therefore:
\begin{equation*}
C(P_k) = \sum_{r=2}^k I^{(r)}(P_k)
\end{equation*}
\end{proof}

\paragraph{Identity Check.}
The core relationships underpinning the decomposition are:
\begin{enumerate}
\item $I^{(1)}(P_k) = - \sum_{i=1}^k H(X_i)$ (by definition of total 1-way interaction)
\item $\sum_{r=1}^k I^{(r)}(P_k) = -H(X_{[k]})$ (by Möbius inversion)
\item $C(P_k) = \sum_{i=1}^k H(X_i) - H(X_{[k]})$ (by definition of total correlation)
\end{enumerate}
Combining (1) and (2) yields $-H(X_{[k]}) = - \sum H(X_i) + \sum_{r=2}^k I^{(r)}(P_k)$, which directly rearranges to $C(P_k) = \sum_{r=2}^k I^{(r)}(P_k)$, confirming Lemma~\ref{lem:TotalCorrelationII}.

\subsection{Main Theorem: KL Decomposition}

\begin{theorem}[Hierarchical KL Decomposition]
\label{thm:KLDecomp}
The KL divergence between the joint distribution $P_k$ and the product reference distribution $Q^{(\otimes k)}$ decomposes exactly as:
\begin{equation} \label{eq:KLDecomp_InteractionForm}
\KL(P_k \| Q^{(\otimes k)}) = \sum_{r=2}^k I^{(r)}(P_k) + \sum_{i=1}^k \KL(P_i \| Q)
\end{equation}

Alternatively, using the identity from Lemma~\ref{lem:TotalCorrelationII}:
\begin{equation} \label{eq:KLDecomp_TotalCorrForm}
\KL(P_k \| Q^{(\otimes k)}) = C(P_k) + \sum_{i=1}^k \KL(P_i \| Q)
\end{equation}
\end{theorem}

\begin{proof}
\textbf{Step 1: Expand the KL Divergence}

Start with the definition of KL divergence (Definition~\ref{def:KLDiv}):
\begin{align*}
\KL(P_k \| Q^{(\otimes k)}) &= \sum_x P_k(x) \log_2 \left[ \frac{P_k(x)}{Q^{(\otimes k)}(x)} \right] \\
                             &= \sum_x P_k(x) \left[ \log_2 P_k(x) - \log_2 \left(\prod_{i=1}^k Q(x_i)\right) \right] \\
                             &= \left[ \sum_x P_k(x) \log_2 P_k(x) \right] - \sum_x P_k(x) \sum_{i=1}^k \log_2 Q(x_i) \\
                             &= -H(X_{[k]}) - \sum_x P_k(x) \sum_{i=1}^k \log_2 Q(x_i) \tag{1}
\end{align*}
The second term can be rewritten by swapping the order of summation and marginalising. Since $\log_2 Q(x_i)$ depends only on $x_i$, we have:
\begin{align*}
\sum_x P_k(x) \sum_{i=1}^k \log_2 Q(x_i) &= \sum_{i=1}^k \sum_x P_k(x) \log_2 Q(x_i) \\
                                         &= \sum_{i=1}^k \sum_{x_i} \left( \sum_{x_{[k]\setminus\{i\}}} P_k(x_i, x_{[k]\setminus\{i\}}) \right) \log_2 Q(x_i) \\
                                         &= \sum_{i=1}^k \sum_{x_i} P_i(x_i) \log_2 Q(x_i) \quad \text{(where $P_i$ is the marginal)}
\end{align*}
Substituting back into Equation (1):
\begin{equation} \label{eq:kl_interm}
\KL(P_k \| Q^{(\otimes k)}) = -H(X_{[k]}) - \sum_{i=1}^k \sum_{x_i} P_i(x_i) \log_2 Q(x_i)
\end{equation}

\textbf{Step 2: Relate to Marginal Divergences}

Recall the definition of the marginal KL divergence (Definition~\ref{def:KLDiv}):
\begin{align*}
\KL(P_i \| Q) &= \sum_{x_i} P_i(x_i) \log_2 \left[ \frac{P_i(x_i)}{Q(x_i)} \right] \\
              &= \sum_{x_i} P_i(x_i) \log_2 P_i(x_i) - \sum_{x_i} P_i(x_i) \log_2 Q(x_i) \\
              &= -H(X_i) - \sum_{x_i} P_i(x_i) \log_2 Q(x_i)
\end{align*}
Rearranging this gives:
\begin{equation*}
- \sum_{x_i} P_i(x_i) \log_2 Q(x_i) = H(X_i) + \KL(P_i \| Q)
\end{equation*}
Summing this expression over all variables $i = 1, \dots, k$:
\begin{equation*}
- \sum_{i=1}^k \sum_{x_i} P_i(x_i) \log_2 Q(x_i) = \sum_{i=1}^k H(X_i) + \sum_{i=1}^k \KL(P_i \| Q)
\end{equation*}

\textbf{Step 3: Substitute and Identify Total Correlation}

Substitute this result back into the expression for the total KL divergence (see Equation~\ref{eq:kl_interm}):
\begin{equation*}
\KL(P_k \| Q^{(\otimes k)}) = -H(X_{[k]}) + \left[ \sum_{i=1}^k H(X_i) + \sum_{i=1}^k \KL(P_i \| Q) \right]
\end{equation*}
Rearranging terms:
\begin{align*}
\KL(P_k \| Q^{(\otimes k)}) &= \left[ \sum_{i=1}^k H(X_i) - H(X_{[k]}) \right] + \sum_{i=1}^k \KL(P_i \| Q) \\
                            &= C(P_k) + \sum_{i=1}^k \KL(P_i \| Q)
\end{align*}
This yields the second form of the theorem (\ref{eq:KLDecomp_TotalCorrForm}).

\textbf{Step 4: Use Interaction Information Identity}

Finally, substitute the identity $C(P_k) = \sum_{r=2}^k I^{(r)}(P_k)$ from Lemma~\ref{lem:TotalCorrelationII} into (\ref{eq:KLDecomp_TotalCorrForm}):
\begin{equation*}
\KL(P_k \| Q^{(\otimes k)}) = \sum_{r=2}^k I^{(r)}(P_k) + \sum_{i=1}^k \KL(P_i \| Q)
\end{equation*}
This completes the proof of the first form of the theorem (\ref{eq:KLDecomp_InteractionForm}).
\end{proof}

\begin{remark}[Interpretation]
\label{rem:decompinterpretation}
Theorem~\ref{thm:KLDecomp} provides a precise decomposition of the total KL divergence into two fundamentally different types of contributions:
\begin{enumerate}
    \item \textbf{Sum of Marginal Divergences}: $\sum_{i=1}^k \KL(P_i \| Q)$. This term quantifies the divergence attributable solely to the deviation of each individual variable's marginal distribution $P_i$ from the reference distribution $Q$. If all marginals match the reference ($P_i = Q$ for all $i$), this term is zero.
    \item \textbf{Sum of Higher-Order Interactions / Total Correlation}: $\sum_{r=2}^k I^{(r)}(P_k) = C(P_k)$. This term quantifies the divergence attributable to the statistical dependencies among the variables $X_1, \dots, X_k$ within the joint distribution $P_k$. It represents the information gain from using the true joint distribution $P_k$ compared to using the product of its marginals $\prod P_i$, adjusted for the reference $Q$. If the variables are independent in $P_k$ (i.e., $P_k = \prod P_i$), then $C(P_k) = 0$ and this term vanishes.
\end{enumerate}
This decomposition provides a precise accounting: the total divergence is the sum of divergences attributable solely to individual variable distributions deviating from $Q$, plus the total correlation, which encapsulates all effects arising from statistical dependencies within $P_k$. If marginals match $Q$, divergence is purely due to interactions; if variables are independent in $P_k$, divergence is purely due to marginal mismatches.
\end{remark}

\section{Visual Illustrations of the Decomposition}
\label{sec:visuals}
To visualise the components of the decomposition presented in Theorem~\ref{thm:KLDecomp}, we illustrate Case 1 ($k=3$, symmetric) from our experiments. These visualisations confirm that the decomposition accurately partitions the total divergence and provides a structured view of marginal versus interaction contributions.

\begin{figure}[ht]
    \centering
    \includegraphics[width=0.8\linewidth]{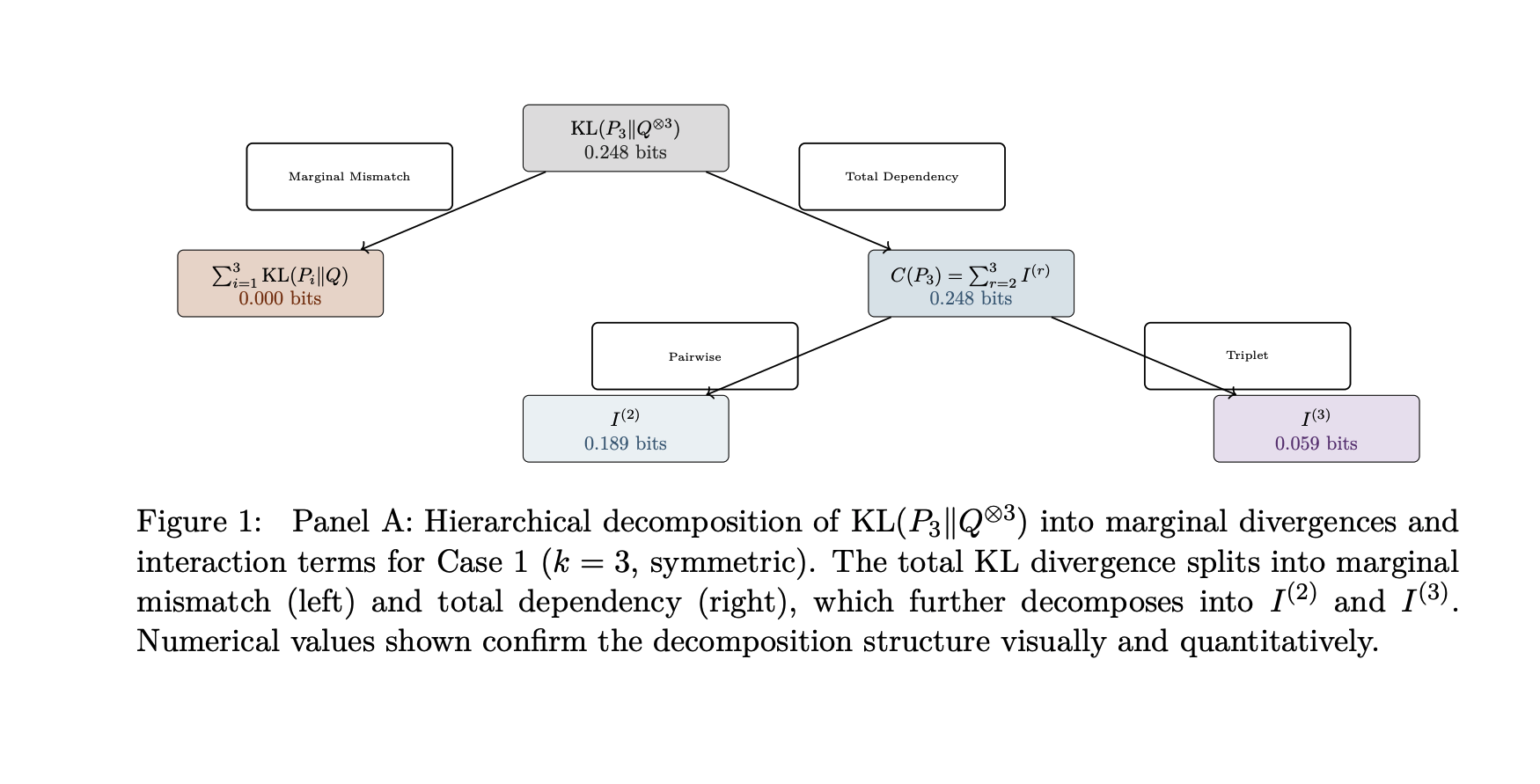}
    \caption{Panel A: Hierarchical decomposition for Case 1 ($k=3$, symmetric).}
    \label{fig:panelA}
\end{figure}

\begin{figure}[ht]
    \centering
    \includegraphics[width=0.8\linewidth]{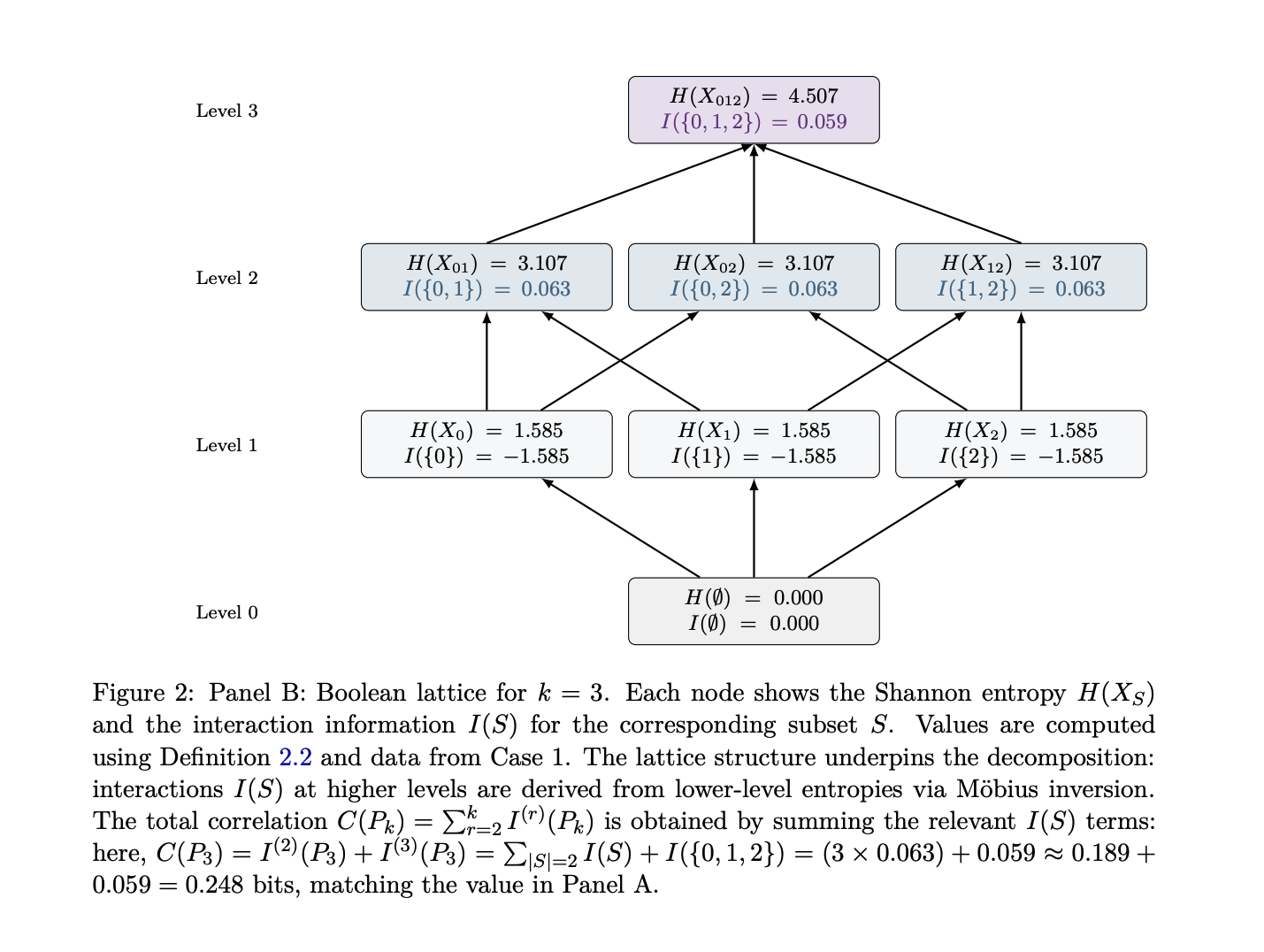}
    \caption{Panel B: Boolean lattice illustration for Case 1.}
    \label{fig:panelB}
\end{figure}

\begin{figure}[ht]
\centering
\includegraphics[width=1\textwidth]{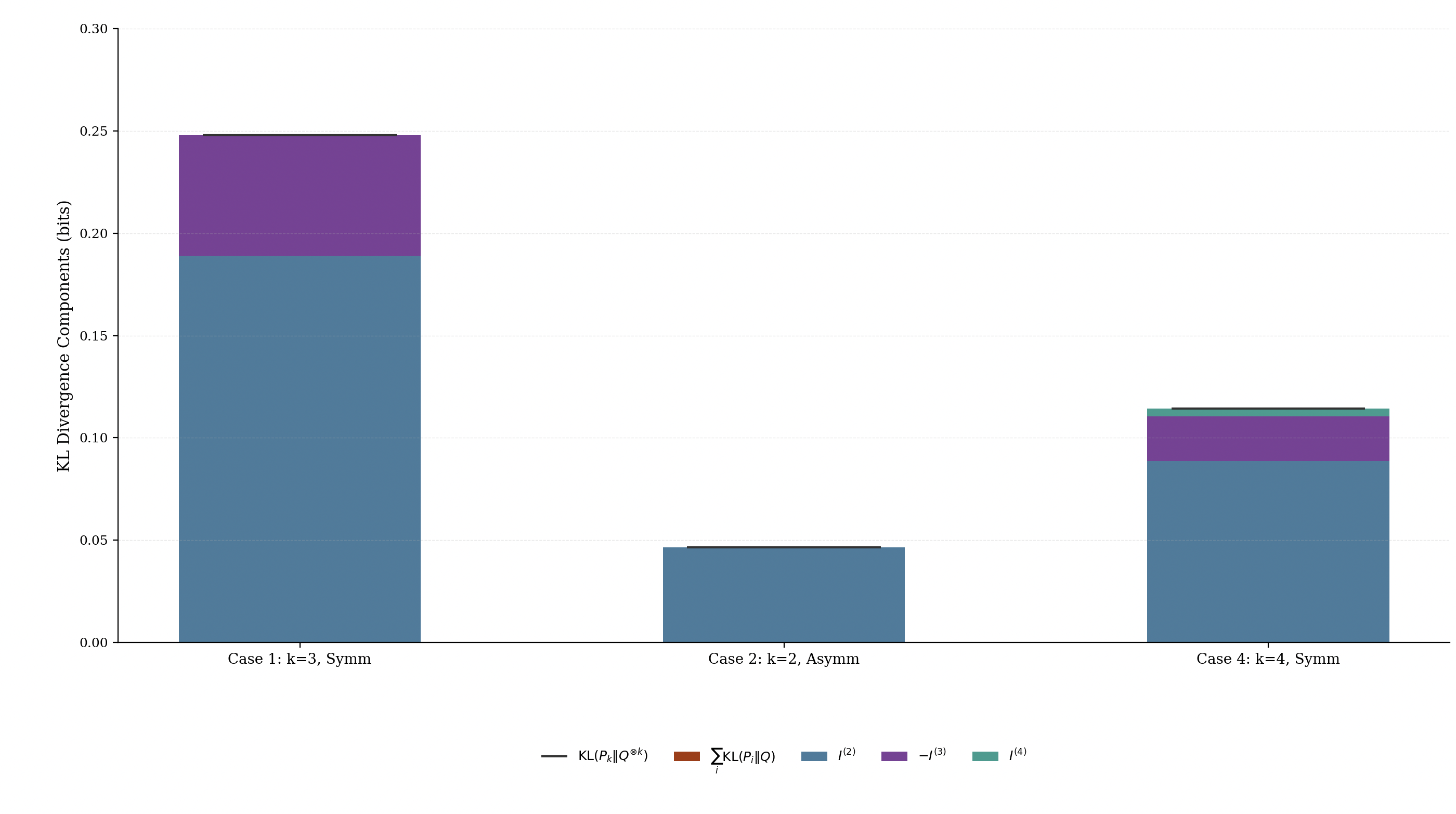}
\caption{Panel C: Empirical validation via stacked bar chart for three test cases (Case 1: $k=3$, symm; Case 2: $k=2$, asymm; Case 4: $k=4$, symm). Bars show the components of the recomposed KL divergence based on Theorem~\ref{thm:KLDecomp}: Sum of marginal KLs $\sum \KL(P_i \| Q)$ (rust red, negligible in these cases as $P_i \approx Q$), total pairwise interaction $I^{(2)}$ (blue), total three-way interaction $I^{(3)}$ (purple), and total four-way interaction $I^{(4)}$ (teal, Case 4 only). The black line marks the independently computed total $\KL(P_k \| Q^{\otimes k})$. The near-perfect match between the stacked bars and the black line (residuals $<10^{-15}$ bits, see Appendix~\ref{app:validationsummary}) numerically validates the exactness of the decomposition theorem across different system sizes and symmetries.}
\label{fig:panelC}
\end{figure}
\clearpage

\section{Experiments and Numerical Validation}
\label{sec:experiments}
We validated Theorem~\ref{thm:KLDecomp} numerically using a multivariate hypergeometric sampling model. This model is advantageous because sampling $k$ items without replacement from a finite population naturally introduces statistical dependencies (typically negative correlations) among the sampled variables $X_1, \dots, X_k$. Crucially, if the population proportions match the reference distribution $Q$, the marginal distribution $P_i$ of any single draw $X_i$ remains exactly $Q$. This setup allows us to isolate and test the interaction terms $I^{(r)}$ of the decomposition, as the marginal KL divergence term $\sum_{i=1}^k \KL(P_i \| Q)$ is theoretically zero (or numerically negligible due to floating-point precision).

We used test cases derived from the data file \texttt{kl\_decomposition\_databaseCORRECTED.json}, including:
\begin{itemize}
    \item \textbf{Case 1: $k=3$, Symmetric Population}
        \begin{itemize}
            \item Population size $n=6$, number of draws $k=3$.
            \item Population counts $Q_{\text{scaled}} = \{0: 2, 1: 2, 2: 2\}$.
            \item Reference distribution $Q = \{0: 1/3, 1: 1/3, 2: 1/3\}$.
        \end{itemize}
    \item \textbf{Case 2: $k=2$, Asymmetric Population}
        \begin{itemize}
            \item Population size $n=5$, number of draws $k=2$.
            \item Population counts $Q_{\text{scaled}} = \{0: 3, 1: 2\}$.
            \item Reference distribution $Q = \{0: 0.6, 1: 0.4\}$.
        \end{itemize}
    \item \textbf{Case 4: $k=4$, Symmetric Binary Population}
        \begin{itemize}
            \item Population size $n=8$, number of draws $k=4$.
            \item Population counts $Q_{\text{scaled}} = \{\text{'A'}: 4, \text{'B'}: 4\}$.
            \item Reference distribution $Q = \{\text{'A'}: 0.5, \text{'B'}: 0.5\}$.
        \end{itemize}
\end{itemize}

\begin{remark}[Vanishing Marginal Divergences in Symmetric Sampling]
In the multivariate hypergeometric setup, when population proportions exactly match the reference distribution $Q$, each marginal $P_i$ satisfies $P_i = Q$ due to the symmetry of the sampling process. This ensures that
\begin{equation*}
\KL(P_i \| Q) = 0 \quad \text{for all } i = 1, \dots, k,
\end{equation*}
so that the entire KL divergence arises from statistical dependencies, i.e., the interaction structure in $P_k$.
\end{remark}

For each case, the validation involved the following computational steps:
\begin{enumerate}
    \item Computing the exact joint probability distribution $P_k$ using multivariate hypergeometric probabilities.
    \item Calculating the product reference distribution $Q^{(\otimes k)}$.
    \item Computing the total KL divergence $\KL(P_k \| Q^{(\otimes k)})$ directly from the distributions.
    \item Calculating all necessary subset marginal distributions $P_S$ and their Shannon entropies $H(X_S)$.
    \item Computing all interaction information terms $I(S)$ using Definition~\ref{def:interactioninfo} and summing them to get the total $r$-way interactions $I^{(r)}(P_k)$.
    \item Computing the sum of marginal divergences $\sum_{i=1}^k \KL(P_i \| Q)$ (expected to be near zero).
    \item Comparing the directly computed $\KL(P_k \| Q^{(\otimes k)})$ with the recomposed value from Theorem~\ref{thm:KLDecomp}: $\sum_{i=1}^k \KL(P_i \| Q) + \sum_{r=2}^k I^{(r)}(P_k)$.
\end{enumerate}

\textbf{Results:} In all tested cases, the residual difference between the directly computed KL divergence and the value recomposed using the decomposition theorem was less than $10^{-15}$ bits, confirming the exactness of the decomposition to machine precision. Key numerical results are summarised below (see Appendix~\ref{app:validationsummary} for details):
\begin{itemize}
    \item \textbf{Case 1}: $\KL_{\text{direct}} \approx 0.247997$, $\sum_{r=2}^3 I^{(r)} \approx 0.247997$. Residual $\approx 1.67 \times 10^{-16}$. The decomposition correctly attributes the divergence entirely to interactions ($I^{(2)} \approx 0.189$, $I^{(3)} \approx 0.059$).
    \item \textbf{Case 2}: $\KL_{\text{direct}} \approx 0.046439$, $I^{(2)} \approx 0.046439$. Residual $\approx 6.94 \times 10^{-17}$. For $k=2$, the decomposition reduces to $\KL = C(P_k) = I^{(2)}$, matching the standard result for mutual information.
    \item \textbf{Case 4}: $\KL_{\text{direct}} \approx 0.114412$, $\sum_{r=2}^4 I^{(r)} \approx 0.114412$. Residual $\approx 5.69 \times 10^{-16}$. This case reveals a hierarchy: pairwise interactions ($I^{(2)} \approx 0.0886$) dominate, but triplet interactions ($I^{(3)} \approx -0.0219$, indicating synergy) and the quadruplet interaction ($I^{(4)} \approx 0.0039$) make non-negligible contributions.
\end{itemize}

These results empirically confirm the validity of Theorem~\ref{thm:KLDecomp}. Figure~\ref{fig:panelC} visually summarises the decomposition for these cases. Furthermore, the identity $C(P_k) = \sum_{r=2}^k I^{(r)}(P_k)$ (Lemma~\ref{lem:TotalCorrelationII}) was independently verified by comparing the sum of interactions to the directly computed total correlation $C(P_k) = \sum H(X_i) - H(X_{[k]})$, with residuals also below machine precision ($< 10^{-15}$ bits), confirming the internal consistency of the framework.

\section{Intuition and Analogy: Understanding the KL Decomposition}
\label{sec:intuition}

When comparing a complex real-world system to a simplified model, a fundamental question arises: \textbf{``When our real system deviates from our simplified expectation, where exactly does this divergence originate?''}

Consider a multivariate system described by a joint probability distribution $P_k(X_1,\dots,X_k)$ — perhaps:
\begin{itemize}
    \item User interactions across features in a recommendation system.
    \item Gene expression patterns in a biological network.
    \item Fluctuations of multiple assets in a financial portfolio.
\end{itemize}

We compare this real system ($P_k$) to a baseline model $Q^{(\otimes k)}$. This baseline makes two strong assumptions about the system's behaviour:
\begin{enumerate}
    \item \textbf{Individual Behaviour}: Each variable $X_i$ follows the same reference distribution $Q$.
    \item \textbf{Independence}: All variables act independently of one another.
\end{enumerate}

The Kullback-Leibler (KL) divergence, $\KL(P_k \| Q^{(\otimes k)})$, quantifies the overall mismatch or ``surprise'' in observing the real system $P_k$ when we expected the simple, independent baseline $Q^{(\otimes k)}$. A higher KL divergence signifies a greater deviation from the baseline expectation.

Our decomposition, presented in Theorem~\ref{thm:KLDecomp}, partitions this single divergence value into precisely two interpretable components, corresponding directly to violations of each baseline assumption:

\subsubsection{Marginal Mismatches — Violations of Assumption (1)}
The first component is the sum of marginal KL divergences:
\begin{equation*}
\sum_{i=1}^k \KL(P_i \| Q)
\end{equation*}
This term measures how much each individual variable $X_i$, viewed in isolation, deviates from the reference distribution $Q$. It quantifies the divergence attributable solely to individual components behaving differently than expected under the baseline, irrespective of any interactions between them.

In our experimental validation using hypergeometric sampling (e.g., Case 1, Figure~\ref{fig:panelC}), this term was approximately zero. This was by design, as the sampling method ensured that individual marginal distributions ($P_i$) matched the reference ($Q$) exactly, thereby isolating the effects of interactions.

\subsubsection{Interaction Structure (Total Correlation) — Violations of Assumption (2)}
The second component is the total correlation, $C(P_k)$:
\begin{equation*}
C(P_k) = \sum_{r=2}^k I^{(r)}(P_k)
\end{equation*}
This term captures the entire contribution to divergence arising from statistical dependencies between variables—relationships that violate the baseline's independence assumption (Assumption 2). Even if every variable individually matches the reference distribution perfectly ($\sum \KL(P_i \| Q) = 0$), their joint behaviour within $P_k$ may reveal complex interdependencies absent in the independent baseline $Q^{(\otimes k)}$.

As shown theoretically (Lemma~\ref{lem:TotalCorrelationII}) and numerically, the total correlation $C(P_k)$ itself decomposes hierarchically. It is the sum of contributions from layers of increasingly complex interactions:
\begin{itemize}
    \item \textbf{Pairwise interactions} ($I^{(2)}$): Captures how pairs of variables influence each other, deviating from pairwise independence.
    \item \textbf{Triplet interactions} ($I^{(3)}$): Measures emergent dependencies visible only when considering three variables together, beyond what pairwise interactions explain.
    \item \textbf{Higher-order interactions} ($I^{(r)}$ for $r > 3$): Accounts for dependencies involving groups of four or more variables.
\end{itemize}

In our experiments (summarised visually in Figure~\ref{fig:panelC}), we observed how these interaction terms ($r \ge 2$) accounted for essentially all the divergence in the hypergeometric cases. Pairwise interactions ($I^{(2)}$, blue sections) often contributed most significantly, but higher-order terms ($I^{(3)}$ purple, $I^{(4)}$ teal) were also non-negligible, demonstrating the importance of capturing complex dependency structures.

\subsubsection{Analogy: Diagnosing an Orchestra's Performance}
Consider an orchestra performing a symphony. The reference baseline model $Q^{(\otimes k)}$ is analogous to expecting:
\begin{enumerate}
    \item Each musician ($X_i$) plays their part precisely according to a standard interpretation ($Q$).
    \item Every musician plays independently, without coordinating or reacting to others.
\end{enumerate}

When we listen to the actual performance ($P_k$), our decomposition helps diagnose deviations from this simplistic baseline:

\textbf{Marginal Mismatches}: This corresponds to individual musicians deviating from the standard interpretation ($Q$) for their part. Perhaps the shy oboist plays consistently too pianissimo, or a violinist's intonation is slightly flat throughout. These are errors traceable to individual performers, regardless of the ensemble.

\textbf{Interaction Structure}: This captures deviations from the unrealistic ``independent play'' assumption. Even if every musician executes their part perfectly when considered alone, the ensemble dynamics create dependencies:
\begin{itemize}
    \item \textbf{Pairwise interactions ($I^{(2)}$)}: The string section's timing relative to the percussion section might be slightly misaligned. The lead violin might adjust their dynamics based on the cello's phrasing.
    \item \textbf{Higher-order interactions ($I^{(3)}$, $I^{(4)}$, ...)}: Complex textures and coordinated swells emerge from entire sections (e.g., woodwinds and brass) resonating together in ways that cannot be reduced to simple pairwise adjustments. The overall blend and balance achieved by the conductor represent systemic dependencies.
\end{itemize}

This analogy mirrors our numerical findings. For instance, in Case 4 (Figure~\ref{fig:panelC}), the total divergence (0.114 bits) wasn't just due to pairwise issues (0.089 bits), but also included significant contributions from triplet (-0.022 bits, indicating synergy perhaps) and quadruplet (0.004 bits) interactions within the ``ensemble'' of variables.

\subsubsection{Diagnostic Value Across Domains}
This decomposition provides numerically tractable and precise diagnostic insight when analysing complex systems against independent baselines:
\begin{itemize}
    \item If marginal mismatches dominate => focus on improving the modelling of individual components or understanding univariate deviations.
    \item If interaction terms dominate (as in our validation cases) => focus on capturing dependencies, network structures, or coordination mechanisms within the system.
\end{itemize}

Unlike heuristic approximations, this framework offers a mathematically exact partition of the KL divergence—the sum of the components equals the total KL divergence to machine precision, as confirmed in our numerical experiments (residuals $< 10^{-15}$ bits, see Appendix~\ref{app:validationsummary}).

Whether analysing model misfit in machine learning, network perturbations in systems biology, or equilibrium deviations in economics, this decomposition identifies not just that a system differs from expectation, but precisely how and where in its marginal versus interaction structure that difference resides.

\section{Discussion}
\label{sec:discussion}
The decomposition derived in Theorem~\ref{thm:KLDecomp}, $\KL(P_k \| Q^{(\otimes k)}) = \sum_{i=1}^k \KL(P_i \| Q) + C(P_k)$, provides an exact and interpretable partition of the total KL divergence into components reflecting marginal mismatches and statistical dependencies.

\textbf{Interpretation and Diagnostic Value:}
The decomposition offers a clear diagnostic lens.
\begin{itemize}
    \item \textbf{Separation of Effects}: It precisely distinguishes divergence arising from local, single-variable properties (the marginal distributions $P_i$ differing from the reference $Q$) from divergence arising from global, systemic properties (the statistical dependencies captured by total correlation $C(P_k)$). A large $\sum \KL(P_i \| Q)$ indicates divergence primarily driven by individual variable mismatches, suggesting focus on refining marginal models or understanding univariate deviations. Conversely, a large $C(P_k)$ points to divergence stemming from statistical dependencies not present in the independent reference $Q^{(\otimes k)}$, suggesting the need to model or understand these interactions.
    \item \textbf{Role of Total Correlation}: The decomposition highlights the fundamental role of total correlation $C(P_k)$ as the component of KL divergence attributable to dependencies. When the marginals match the reference ($P_i = Q$), the decomposition simplifies to $\KL(P_k \| Q^{(\otimes k)}) = C(P_k)$, directly equating the divergence to the total dependency within the system relative to the independent state implicitly defined by $Q^{(\otimes k)}$.
    \item \textbf{Hierarchical Insight}: By further decomposing $C(P_k)$ into the sum of $r$-way interaction terms $\sum_{r=2}^k I^{(r)}(P_k)$, the framework allows investigating the complexity of these dependencies. As seen in Case 4, one can assess the relative contributions of pairwise ($I^{(2)}$), triplet ($I^{(3)}$), and higher-order interactions to the overall divergence. This reveals whether simple pairwise models suffice or if higher-order structure is crucial.
\end{itemize}

\textbf{Applications:}
The ability to precisely attribute divergence offers benefits across various fields:
\begin{itemize}
    \item \textbf{Machine Learning}: Beyond simple feature importance, this decomposition can guide model selection and architecture design. It quantifies the trade-off between fitting marginal distributions accurately versus capturing complex interactions (e.g., choosing between simpler additive models, models with pairwise interaction terms, or highly non-linear models like deep networks). It can also diagnose sources of error or divergence in generative models, variational autoencoders (in the ELBO objective), and reinforcement learning policy evaluation.
    \item \textbf{Complex Systems}: In neuroscience, the decomposition can differentiate contributions to neural coding divergence arising from changes in single-neuron firing rate statistics (marginals) versus changes in population synchrony, synergy, or higher-order correlations (interactions) \cite{Schneidman2006}. In genomics or systems biology, it can help distinguish the effects of individual gene dysregulation from altered pathway interactions or epistatic effects in disease models compared to healthy baselines.
    \item \textbf{Information Geometry}: The decomposition provides insight into the geometric structure of the space of probability distributions. It relates the KL divergence (a Bregman divergence) between a point $P_k$ and a reference product measure $Q^{(\otimes k)}$ to the total correlation (related to the divergence to the product of marginals $\prod P_i$) and the sum of marginal divergences. This enhances understanding of how dependencies shape the information manifold \cite{Amari2016}.
\end{itemize}

\textbf{Limitations and Future Directions:}
Several aspects warrant further investigation:

\paragraph{Extension: General Product References.}
While this paper focuses on reference distributions of the form $Q^{(\otimes k)}$ (identical marginals), the decomposition should readily generalise to arbitrary product references
\begin{equation*}
Q_{1:k}(x_1, \dots, x_k) := \prod_{i=1}^k Q_i(x_i),
\end{equation*}
where $Q_i$ may differ across dimensions. In this case, the decomposition becomes:
\begin{equation*}
\KL(P_k \| Q_{1:k}) = \sum_{i=1}^k \KL(P_i \| Q_i) + C(P_k),
\end{equation*}
with total correlation $C(P_k)$ defined as before. This variant may be particularly useful in applications where the baseline marginal distributions are heterogeneous.

\begin{itemize}
    \item \textbf{Continuous Variables}: The current derivation relies on Shannon entropy for discrete variables. Extending the framework rigorously to continuous systems requires careful treatment of differential entropies, reference measures (e.g., Lebesgue measure), and potential infinities, perhaps using tools like relative entropy, copula theory, or kernel density estimation. Extensions to continuous domains require non-trivial reinterpretations of entropy (differential entropy is not a Möbius-invertible set function).
    \item \textbf{Computational Complexity}: The number of subsets involved in calculating all interaction information terms grows exponentially ($\mathcal{O}(2^k)$) with the number of variables $k$. Exact computation is thus infeasible for very high-dimensional systems. Practical applications for large $k$ will likely require computationally efficient approximations, perhaps based on assuming sparse interactions (e.g., graphical models where only interactions corresponding to cliques are non-zero) or focusing only on lower-order interactions (e.g., truncating the sum at $r=2$ or $r=3$).
    \item \textbf{Estimation from Data}: Estimating high-order entropies and interaction information accurately from finite samples is statistically challenging due to the curse of dimensionality. Robust estimation techniques, potentially incorporating Bayesian methods with appropriate priors or regularisation, are needed for reliable application to empirical data \cite{Presse2013}.
\end{itemize}


Our decomposition utilises M\"obius inversion on the subset lattice, a tool also employed in the related work by Jansma et al.~\cite{Jansma2024} for decomposing KL divergence. However, the approaches and resulting frameworks differ significantly. Our work focuses specifically on decomposing $KL(P_k \| Q^{(\otimes k)})$ -- the divergence from an independent reference with identical marginals $Q$. We achieve this by separating the total divergence into two primary, interpretable components: the sum of first-order marginal KL divergences ($\sum_{i=1}^k KL(P_i \| Q)$), which isolates deviations of individual variables from the reference $Q$, and the total correlation ($C(P_k)$). We further show that $C(P_k)$, which captures all statistical dependencies within $P_k$, decomposes hierarchically into contributions from $r$-way interactions ($\sum_{r=2}^k I^{(r)}(P_k)$) based on the M\"obius inversion of Shannon entropy (Lemma~\ref{lem:TotalCorrelationII}). In contrast, the framework developed by Jansma et al.~\cite{Jansma2024} addresses the general multivariate $KL(p\|q)$ by applying M\"obius inversion \emph{directly} to the marginal KL divergences ($KL(p_S\|q_S)$) across all subsets $S$. This yields $2^k$ signed components ($\Delta_{p\|q}(S)$), each reflecting a contribution associated with a specific subset $S$. While their method offers a fine-grained decomposition applicable to general KL divergence based on the lattice structure of marginal KLs, our decomposition (Theorem~\ref{thm:KLDecomp}) provides a distinct structure centered on the interpretable separation between collective first-order marginal effects and the total, hierarchically structured, statistical dependency ($C(P_k)$). This comparison highlights the specific contribution of our approach in providing this particular additive structure using standard Shannon information quantities.

\section{Conclusion}
\label{sec:conclusion}

We have derived and empirically validated an algebraically exact and hierarchical decomposition of the Kullback-Leibler divergence $\KL(P_k \| Q^{(\otimes k)})$ into additive components consisting solely of the sum of marginal divergences and the total correlation. This decomposition is mathematically complete—accounting for the entire divergence with no residual—and interpretable, distinguishing marginal mismatches from the structural dependencies within the joint distribution. The total correlation term is further decomposed into higher-order interaction information terms via Möbius inversion, offering a layered view of dependency structure.

This decomposition deepens our theoretical understanding of KL divergence and its relationship to fundamental information-theoretic quantities like entropy and interaction information. More practically, it offers a powerful diagnostic tool for analysing complex systems across various domains, from machine learning to computational biology and neuroscience, by pinpointing whether divergence stems from individual component behaviour or systemic interactions.

Future research directions include developing computationally efficient approximations for high-dimensional systems, potentially leveraging techniques from graphical modelling or sparse interaction recovery; extending the framework rigorously to continuous variables, perhaps using copula-based methods or relative differential entropy; and creating software tools and visualisation techniques to facilitate the application of this decomposition in data analysis workflows across scientific domains.

\section*{Author's Note}
{\small
I am a second-year undergraduate in economics at the University of Bristol working independently. This paper belongs to a research programme I am developing across information theory, econometrics and mathematical statistics.

This work was produced primarily through AI systems that I directed and orchestrated. The AI generated the mathematical content, proofs and symbolic derivations based on my research questions and guidance. I have no formal mathematical training but am eager to learn through this process of directing AI-powered mathematical exploration.

My contribution involves designing research directions, evaluating and selecting AI outputs, and ensuring the coherence of the overall research agenda. All previously published work that influenced these results has been cited to the best of my knowledge. The presentation aims to be pedagogically accessible.
}


\appendix
\section{Numerical Validation Data Summary}
\label{app:validationsummary}

This appendix summarizes the key numerical values obtained during the validation process described in Section~\ref{sec:experiments}. Here, $\mathrm{KL}_{\mathrm{full}}$ is the directly computed $\KL(P_k \| Q^{(\otimes k)})$, $\mathrm{KL}_{\mathrm{marginals\_sum}}$ is $\sum \KL(P_i \| Q)$, $\mathrm{TotalCorrelation}_{C(P_k)}$ is the sum $\sum_{r=2}^k I^{(r)}$ from the decomposition, and $\mathrm{Direct}_{C(P_k)}$ is the total correlation calculated as $\sum H(X_i) - H(X_{[k]})$. The $\mathrm{Residual}$ compares $\mathrm{KL}_{\mathrm{full}}$ to the sum $\mathrm{KL}_{\mathrm{marginals\_sum}} + \mathrm{TotalCorrelation}_{C(P_k)}$. The term $I_{\mathrm{sums}}$ shows the values for each total $r$-way interaction $I^{(r)}$. All values are in bits.

The numerical data is presented below:

\begin{verbatim}
--- Database Summary ---
Case: Case1_k3_Symm
  KL_full: 0.24799690655495005
  KL_marginals_sum: 0.000000e+00
  TotalCorrelation_C_Pk: 0.24799690655495005  # Sum I^(r) for r>=2
  Direct_C_Pk: 0.24799690655495027            # Sum H(X_i) - H(X_{[k]})
  Residual (KL vs Decomp): 1.665335e-16
  I_sums: {'2': 0.18910321749875303, '3': 0.05889368905619702}

Case: Case2_k2_Asymm
  KL_full: 0.04643934467101547
  KL_marginals_sum: 0.000000e+00
  TotalCorrelation_C_Pk: 0.04643934467101547  # Sum I^(r) for r>=2 (only I^(2))
  Direct_C_Pk: 0.04643934467101547            # Sum H(X_i) - H(X_{[k]})
  Residual (KL vs Decomp): 6.938894e-17
  I_sums: {'2': 0.04643934467101547}

Case: Case4_k4_Symm
  KL_full: 0.11441198342591395
  KL_marginals_sum: 6.406853e-16              # Numerically non-zero but negligible
  TotalCorrelation_C_Pk: 0.1144119834259133   # Sum I^(r) for r>=2
  Direct_C_Pk: 0.11441198342591374            # Sum H(X_i) - H(X_{[k]})
  Residual (KL vs Decomp): 5.689893e-16       # Slightly higher due to KL_marginals_sum
  I_sums: {'2': 0.08863118375487466, '3': -0.02188937283553888,
           '4': 0.003891426907205896}
\end{verbatim}

Note: Minor differences between $\mathrm{TotalCorrelation}_{C(P_k)}$ and $\mathrm{Direct}_{C(P_k)}$ (around $10^{-16}$) arise from floating-point arithmetic differences in the calculation paths (summing interactions versus summing entropies). The $\mathrm{Residual}$ confirms that the main decomposition theorem holds.
\end{document}